\documentclass{amsart}
\usepackage{amsmath,amssymb}
\usepackage{color}
\usepackage{bm}
\usepackage{enumerate}
\usepackage{graphicx}
\usepackage[numbers]{natbib}
\usepackage{enumerate}
\usepackage{hyperref}

\newcommand{\ba}{\mathbf{a}}
\newcommand{\be}{\mathbf{e}}
\newcommand{\bbf}{\mathbf{f}}
\newcommand{\bu}{\mathbf{u}}
\newcommand{\bv}{\mathbf{v}}
\newcommand{\bx}{\mathbf{x}}
\newcommand{\bc}{\mathbf{c}}
\newcommand{\bn}{\mathbf{n}}
\newcommand{\by}{\mathbf{y}}
\newcommand{\bw}{\mathbf{w}}
\newcommand{\bz}{\mathbf{z}}
\newcommand{\bA}{\mathbf{A}}
\newcommand{\bC}{\mathbf{C}}

\newcommand{\bE}{\mathbf{E}}

\newcommand{\bG}{\mathbf{G}}
\newcommand{\bB}{\mathbf{B}}
\newcommand{\bI}{\mathbf{I}}
\newcommand{\bK}{\mathbf{K}}
\newcommand{\bM}{\mathbf{M}}
\newcommand{\bU}{\mathbf{U}}

\newcommand{\bR}{\mathbf{R}}
\newcommand{\bV}{\mathbf{V}}
\newcommand{\bW}{\mathbf{W}}

\newcommand{\bzero}{\mathbf{0}}

\newcommand{\btW}{\widetilde{\mathbf{W}}}
\newcommand{\btA}{\widetilde{\mathbf{A}}}
\newcommand{\btE}{\widetilde{\mathbf{E}}}
\newcommand{\btK}{\widetilde{\mathbf{K}}}

\newcommand{\btc}{\widetilde{\mathbf{c}}}

\newcommand{\tq}{\widetilde{q}}

\newcommand{\hbbf}{\widehat{\mathbf{f}}}
\newcommand{\hbu}{\widehat{\mathbf{u}}}
\newcommand{\hbW}{\widehat{\mathbf{W}}}

\newcommand{\cL}{\mathcal{L}}
\newcommand{\cO}{\mathcal{O}}

\newcommand{\om}{\omega}
\newcommand{\eps}{\epsilon}

\newcommand{\real}{\mathbb{R}}

\newcommand{\rank}{\mathrm{rank}\,}
\newcommand{\linspan}{\mathrm{span}\,}
\newcommand{\diag}{\mathrm{diag}\,}
\newcommand{\range}{\mathcal{R}}
\newcommand{\nullspace}{\mathcal{N}}

\newcommand{\cross}{\wedge}

\newcommand{\M}[1]{\left({#1}\right)}

\newcommand{\Mcb}[1]{\left\{{#1}\right\}}
\newcommand{\norm}[1]{\left\|{#1}\right\|}

\newcommand{\dsp}{\displaystyle}

\usepackage{amsthm}
\newtheorem{lemma}{Lemma}
\newtheorem{theorem}{Theorem}
\newtheorem{remark}{Remark}

\newtheorem{definition}{Definition}

\newcommand{\figref}[1]{Figure~\ref{#1}}
\newcommand{\lemref}[1]{Lemma~\ref{#1}}
\newcommand{\remref}[1]{Remark~\ref{#1}}
\newcommand{\thmref}[1]{Theorem~\ref{#1}}
\newcommand{\secref}[1]{\S\ref{#1}}
\newcommand{\appref}[1]{Appendix~\ref{#1}}

\providecommand{\doi}[1]{\href{http://dx.doi.org/#1}{doi:#1}}

\begin{document}
\title[The response of elastodynamic networks]{Complete
characterization and synthesis of the\\ response function of elastodynamic
networks}
\author[F. Guevara Vasquez]{Fernando Guevara Vasquez}
\address{Mathematics Department, University of Utah, 155 S 1400 E Rm.
233, 84112 Salt Lake City, Utah}
\email{fguevara@math.utah.edu}

\author[G. W. Milton]{Graeme W. Milton}
\email{milton@math.utah.edu}

\author[D. Onofrei]{Daniel Onofrei}
\email{onofrei@math.utah.edu}

\begin{abstract}
The response function of a network of springs and masses, an
elastodynamic network, is the matrix valued function $\bW(\om)$,
depending on the frequency $\om$, mapping the displacements of some
accessible or terminal nodes to the net forces at the terminals. We give
necessary and sufficient conditions for a given function $\bW(\om)$ to
be the response function of an elastodynamic network, assuming there is
no damping. In particular we construct an elastodynamic network that can
mimic a suitable response in the frequency or time domain.  Our
characterization is valid for networks in three dimensions and also for
planar networks, which are networks where all the elements,
displacements and forces are in a plane. The network we design can fit
within an arbitrarily small neighborhood of the convex hull of the
terminal nodes, provided the springs and masses occupy an arbitrarily
small volume. Additionally, we prove stability of the network response
to small changes in the spring constants and/or addition of springs with
small spring constants.
\end{abstract}
\maketitle


\section{Introduction}
\label{sec:intro}

Is it possible to design an elastic material that has a prescribed
response?  This question is answered by Camar-Eddine and Seppecher
\cite{camareddine_seppecher2003} for linear elastic materials in three
dimensions, assuming the macroscopic response is governed by a single
displacement field. Their approach consists of three steps. First it is
shown how to design a continuum material that behaves like an {\em
elastic network} (a network composed of springs). Then the response of
elastic networks is characterized, i.e. it is shown how to construct an
elastic network with a suitable response. A limiting argument is then
used to answer the question for the continuum. As a first step towards
solving the characterization problem when the response depends on time,
we show how to design an {\em elastodynamic network} (a network of
springs and masses), that can mimic a prescribed response as a function
of time (or frequency).  Moreover if the springs and masses occupy an
arbitrarily small volume, the network can be designed to fit within an
arbitrarily small neighborhood of the convex hull of the terminal nodes,
which is a requirement for an argument similar to that of Camar-Eddine
and Seppecher \cite{camareddine_seppecher2003}. An earlier
characterization of elastodynamic networks is that of Milton and
Seppecher \cite{milton_seppecher2008}. However the network elements used
in the construction \cite{milton_seppecher2008} are frequency dependent,
so the constructed network can only mimic the response function at a
single fixed frequency.

In a different context, the approach of Camar-Eddine and Seppecher was
applied earlier by the same authors \cite{Camar:2002:CSD} to
characterize all possible responses for the conductivity equation,
assuming the macroscopic response is governed by a single voltage field. The
problem of finding a network with a given response is often called
``network synthesis'', and the earliest example is Kirchhoff's
$Y-\Delta$ theorem, which characterizes the response of any resistor
network in three dimensions. Another characterization for resistor
networks is that of Curtis, Ingerman and Morrow
\cite{curtis_ingerman_morrow1998} who consider planar networks that can
be embedded inside a disk and where all terminals lie on its boundary.
For electrodynamic networks (with resistances, capacitors and
inductances), we are only aware of results dealing with the frequency
response or impedance of a circuit with two terminals (see Foster
\cite{Foster:1924:RT,Foster:1924:TRD} and Bott and Duffin
\cite{Bott:1949:ISU}).  Milton and Seppecher \cite{milton_seppecher2008}
give a construction for $n-$terminal elastodynamic, electrodynamic and acoustic
networks which is valid at a single frequency. The electromagnetic
analog of elastodynamic networks is considered by the same authors
\cite{Milton:2009:EC,Milton:2009:HEC}

In \secref{sec:enet} we give the properties of the response function of
elastic and elastodynamic networks. The construction of a network that
matches a response function with the properties in \secref{sec:enet} is
given for the static case in \secref{sec:static}. Note that the
characterization of elastic networks by Camar-Eddine and Seppecher
\cite{camareddine_seppecher2003} is part of a limiting argument on
energy functionals, so only non-degenerate three dimensional elastic
networks are needed. The degenerate case corresponds to {\em planar
elastic networks} (the network, forces and displacements lie on a plane)
and is a set of measure zero which leaves the energy functionals
considered in \cite{camareddine_seppecher2003} unaffected. We complete
the characterization in \cite{camareddine_seppecher2003} to include
planar elastic networks. Then in \secref{sec:dynamic} we completely
characterize the response of elastodynamic networks (planar or in three
dimensions) for all frequencies and assuming there is no dissipation
(damping) in the network. We include in the appendices two technical
results. \appref{app:pert} shows that the network response is stable with
respect to small changes in the spring constants and the addition (but
not deletion) of springs. \appref{app:floppy} uses stability to give a
systematic method of modifying an elastic network to 
eliminate floppy modes without changing significantly the response.
Floppy modes correspond to nodes that can move with zero forces and they
are discussed in more length in \secref{sec:interior}.

\subsection{Preliminaries}
Consider a network composed of springs and masses, and assume we only have
access to $n$ ``terminal'' or ``boundary'' nodes $\bx_1,\ldots,\bx_n \in
\real^d$, where the dimension $d$ is either $2$ or $3$. The network is said to be {\em planar} if $d=2$ and the springs do not cross. The static response
matrix or displacement-to-forces map is the $nd \times nd$ matrix $\bW$
so that $$\bbf = \bW \bu,$$ where $\bu = (\bu_1^T, \ldots, \bu_n^T)^T$ is
the vector of displacements $\bu_i$ of the terminal nodes $\bx_i$ and $\bbf =
(\bbf_1^T, \ldots, \bbf_n^T)^T$ is the vector of net forces $\bbf_i$ acting
on node $\bx_i$ at equilibrium.

In the dynamic case the displacements $\bu(t)$ and $\bbf(t)$ depend on
time $t$. Let $\hbu(\om)$ be the Fourier transform of $\bu(t)$,
\[
 \hbu(\om) = \int_{-\infty}^\infty \bu(t) e^{-i\om t} dt,
\]
and similarly for $\hbbf(\om)$, where $\om$ is the frequency. Then if
$\om$ is not a resonance frequency of the network (a precise definition
of resonance is given later in \secref{sec:dynprop}), the response
matrix of the network is the possibly complex $nd \times nd$ matrix
valued function $\hbW(\om)$ such that
\[
 \hbbf(\om) = \hbW(\om) \hbu(\om).
\]
For convenience we have chosen to work in the frequency domain. However
when $\bu(t) = 0$ for $t<0$, our results can be reformulated for the
{\em transfer function} of the network since $\cL[ \bu(t) ](s) =
\hbu(-is)$, where $\cL$ denotes the Laplace transform, i.e.
\[
 \cL[ \bu(t) ](s) = \int_0^{\infty} \bu(t) e^{-st} dt.
\]
In this case the transfer function of the network is $\hbW(-is)$.
As we work only in the frequency domain,  we drop the hats in the
Fourier transform notation for the sake of clarity (i.e.
$\bu(\om) \equiv \hbu(\om)$ etc$\ldots$). Also as there is no
dissipation, it suffices to assume that $\bu(\om)$ and $\bbf(\om)$ are
real to determine the real valued function $\bW(\om)$.

\section{The response function of an elastodynamic network}
\label{sec:enet}
In this section we establish the properties that the response of an
elastodynamic network satisfies. We start with the response of networks
(static or dynamic) where all the nodes are terminals
(\secref{sec:allterminal}) and then study the case where interior nodes
are present (\secref{sec:interior}). We also include some
transformations in \secref{sec:transf} that do not affect the response
function.

\subsection{Response function for networks without interior nodes}
\label{sec:allterminal}
Consider the simple network consisting of two nodes $\bx_1$
and $\bx_2$ with masses $m_1$ and $m_2$, linked with a spring with
spring constant $k_{1,2}$. Let $\ba_i$  be the force
exerted by the spring on node $\bx_i$, $i=1,2$. By Hooke's law
\[
 \ba_2 = - k_{1,2} \frac{(\bx_2 - \bx_1)(\bx_2 - \bx_1)^T}{\norm{\bx_2 -
 \bx_1}^2}
 (\bu_2 - \bu_1) = -\ba_1.
\]
The laws of motion can be written in matrix form as $ -\omega^2 \bM \bu =
\bbf - \bK \bu$, where
\[
 \begin{aligned}
 &\bK = k_{1,2} \begin{bmatrix} \bn_{1,2} \bn_{1,2}^T & -\bn_{1,2}
 \bn_{1,2}^T\\ -\bn_{1,2} \bn_{1,2}^T & \bn_{1,2} \bn_{1,2}^T
 \end{bmatrix},
 \quad
 \bM = \diag(m_1\be,m_2\be),\\
&\bn_{1,2} = \frac{\bx_2 - \bx_1}{\norm{\bx_2 - \bx_1}},
\end{aligned}
\]
and the vector $\be = (1,\ldots,1)^T \in \real^d$ for $d=2,3$. Thus the
response function of a single spring is given by
\begin{equation}
\label{eq:allterm}
\bW(\om) = \bK - \om^2 \bM.
\end{equation}

When all nodes are terminal nodes (i.e. there are no interior nodes) the
response function can also be written in the form \eqref{eq:allterm}, but
now $\bK$ is the stiffness matrix of the network and the mass matrix $\bM =
\diag(m_1\be,\ldots,m_n\be)$, where $m_i\geq0$ is the mass of the $i-$th
node and $n$ is the number of nodes. The stiffness matrix of the network
is the sum of the stiffness matrices associated with the individual
springs,
\[
 \bK = \sum_{\text{springs}~i,j} [ \bE_i, \bE_j ] \bK_{(i,j)}  [ \bE_i,
 \bE_j]^T \; \in \; \real^{nd \times nd},
\]
where the summation is over all pairs of nodes $\bx_i,\bx_j$ connected
by a spring. The $2d \times 2d$ matrix $\bK_{(i,j)}$ is the stiffness
matrix of the spring between nodes $\bx_i$ and $\bx_j$. We also used $nd
\times d$ matrices $\bE_i = [\be_{(i-1)d+1}, \ldots, \be_{id}]$, which
are introduced so that the components of $\bK_{(i,j)}$ enter the appropriate
blocks of $\bK$ ($\be_p$ is the $p-$th canonical vector in
$\real^{nd}$).  We only consider non-negative spring stiffnesses
$k_{i,j}$.  Stiffnesses with a non-zero imaginary part model damping or
dissipation of energy in the network and are left for future studies.

\subsection{Response function for networks with interior nodes}
\label{sec:interior}

\subsubsection{The static case} The response matrix $\bW$ can be
obtained from the response matrix $\bA$ of the network where all
nodes are considered as terminal nodes. The partitioning of the nodes
into interior nodes $I$ and terminal (boundary) nodes $B$ induces the
following partitioning of $\bA$,
\begin{equation}
 \bA = \begin{bmatrix} \bA_{BB} & \bA_{BI}\\
 \bA_{IB} & \bA_{II} \end{bmatrix}.
 \label{eq:partition}
\end{equation}
Instead of dealing directly with the response matrix $\bA$, it is
convenient to introduce the quadratic form
\[
 q_{\bA}(\bu) = \bu^T \bA \bu,
\]
which represents twice the total elastic energy stored in the network.  In
the simple case of a single spring between nodes $\bx_1$ and $\bx_2$ with
spring constant $k$, the quadratic form is
\[
 s_{(\bx_1,\bx_2)}(\bu_1,\bu_2) = k \M{(\bu_1-\bu_2) \cdot
 \frac{\bx_1 - \bx_2}{\norm{\bx_1 - \bx_2}}}^2.
\]
We omit the spring constant indices for clarity. When there are more
springs $q_\bA$ is the sum of the quadratic forms for all springs, thus
$q_\bA(\bu) \geq 0$.

For general static networks the response matrix is defined indirectly by
its quadratic form $q_{\bW}$:
\begin{equation}
 q_{\bW}(\bu_B) = \inf_{\bu_I} q_{\bA} (\bu_B, \bu_I).
 \label{eq:quad}
\end{equation}
By the partitioning \eqref{eq:partition} we may rewrite
\[
 q_{\bA} (\bu_B,\bu_I) = \bu_B^T \bA_{BB} \bu_B +
 2\bu_B^T \bA_{BI} \bu_I + \bu_I^T \bA_{II} \bu_I.
\]
The first order optimality conditions for the minimization
\eqref{eq:quad} are actually the balance of forces at the interior nodes:
\[
 \bzero = \nabla_{\bu_I} q_{\bA}(\bu_B,\bu_I) =  2\bA_{II}
 \bu_I + 2 \bA_{IB} \bu_B.
\]
The following lemma shows that for any $\bu_{B}$ it is possible to
balance forces at the interior nodes and it implies the minimization
\eqref{eq:quad} has at least a minimizer (since $q_\bA$ is bounded below).
Another way of seeing this lemma is that if there are any ``floppy'' modes
within the interior nodes (i.e. modes that generate displacements with zero
forces) then those modes are not coupled to the terminals.
\begin{lemma}
\label{lem:floppy}
Given the partitioning \eqref{eq:partition} of the response
matrix $\bA$ where all nodes are considered as terminal nodes, we have
$\range(\bA_{IB}) \subset \range(\bA_{II})$. Here $\range(\bB)$ denotes the
range of a matrix $\bB$.
\end{lemma}
\begin{proof}
By reciprocity $\bA^T=\bA$, thus it is equivalent to prove $\nullspace(\bA_{BI}) \supset
\nullspace(\bA_{II})$, where $\nullspace(\bB)$ denotes the nullspace of a
matrix $\bB$. Let $\bu_{I}$ be a displacement such that $\bA_{II} \bu_{I} =
\bzero$ (i.e. a ``floppy'' mode). Then
 \[
  \begin{aligned}
  0 &= \bu_I^T \bA_{II} \bu_{I} 
  = \begin{bmatrix} \bzero & \bu_I^T \end{bmatrix} \bA \begin{bmatrix} \bzero \\\bu_I \end{bmatrix}\\
  &= \sum_{\text{springs} ~  i,\;j \in I} 
  s_{(\bx_i,\bx_j)}(\bu_i,\bu_j) + \sum_{\text{springs} ~ i \in I, j\in B}
  k_{i,j} \M{\bu_i \cdot \frac{\bx_i - \bx_j}{\norm{\bx_i-\bx_j}}}^2.
  \end{aligned}
 \]
 Therefore for all nodes $\bx_i \in I$ and $\bx_j \in B$ that are linked
 by a spring we must have $\bu_i \cdot (\bx_i - \bx_j) = 0$,
 which means precisely that $\bA_{BI} \bu_I = \bzero$.
\end{proof}

\begin{remark}
We show later in \appref{app:floppy} that floppy modes can be eliminated
from a network by adding springs with small spring constants. The
response of the new network can be made arbitrarily close to that of the
original one, provided the new springs have sufficiently small
stiffness. Examples of floppy modes are given in \figref{fig:floppy}.
\end{remark}

By eliminating the interior nodes, the static response matrix can thus
be written in Schur complement form: 
\begin{equation}
 \bW = \bA_{BB} - \bA_{BI} \bA_{II}^\dagger \bA_{IB},
 \label{eq:schur}
\end{equation}
where $\dagger$ stands for the Moore-Penrose pseudo-inverse, which is
is simply the inverse if there are no floppy modes.

We denote by $\bu \cross \bv$ the cross product of the vectors $\bu,\bv
\in \real^d$. For $d=2$ we have $\bu \cross \bv = \det[\bu,\bv]$ and
for $d=3$, $\bu \cross \bv = (u_2 v_3 - u_3 v_2, u_3 v_1 - u_1 v_3, u_1
v_2 - u_2 v_1)^T$. Before reviewing some properties of the static
response matrix we need the following definition.
\begin{definition}
A {\em balanced system of forces} $\bbf_i$, $i=1,\ldots,n$ supported at
nodes $\bx_i$, $i=1,\ldots,n$ in $\real^d$ ($d=2,3$) satisfies:
 \begin{enumerate}[(a)]
  \item $\dsp\sum_{i=1}^n \bbf_i = \bzero$ (balance of forces)
  \item $\dsp\sum_{i=1}^n \bx_i \cross \bbf_i = \bzero$ (balance of torques)
 \end{enumerate}
\end{definition}
\begin{lemma}
\label{lem:static}
The static response matrix satisfies the following properties.
\begin{enumerate}[(a)]
 \item $\bW \; \in \; \real^{nd \times nd}$.
 \item $\bW = \bW^T$ (reciprocity)
 \item $\bW$ is positive semidefinite (energy is not produced
 by the network)
 \item Every column $\bbf = (\bbf_1^T,\ldots,\bbf_n^T)^T$ of $\bW$ is a
 {\em balanced system of forces} when supported at the nodes $\bx_i$ in
 $\real^d$.
\end{enumerate}
\end{lemma}
\begin{proof}
 Properties (a), (b) and (d) follow from the construction of the
 response matrix. We now prove Property (c). Let $\btW$ be the
 response matrix of a network if all the nodes are considered as
 terminal nodes. Then for all displacements $\bu \in \real^{nd}$ we have
 \[
  \tq(\bu) = \bu^T \btW \bu = \sum_{\text{springs}~i,j} s_{(\bx_i,\bx_j)}
 (\bu_i,\bu_j) \geq 0,
 \]
 where $s_{(\bx_i,\bx_j)} (\bu,\bv)$ is the quadratic form associated with
 the spring between nodes $\bx_i$ and $\bx_j$.  Thus (c) holds
 for networks where all the nodes are terminals. Using \eqref{eq:quad}
 we see that (c) holds for general networks as well.
\end{proof}

\subsubsection{The dynamic case}
\label{sec:dynprop}
The response function in the dynamic case can be obtained in a similar way
as in the static case. First if all the nodes are terminal nodes, the
response function $\bA(\om)$ of the network is given by
\eqref{eq:allterm}. The partitioning of $\bA$ induced by the partitioning
of the nodes into boundary $B$ and interior $I$ nodes is,
\[
 \bA(\om) = 
  \begin{bmatrix} \bK_{BB} & \bK_{BI}\\ \bK_{IB} & \bK_{II} \end{bmatrix} 
  - \om^2 \begin{bmatrix} \bM_{BB} & \\ & \bM_{II} \end{bmatrix}.
\]
As in the static case we can introduce the quadratic form
\begin{equation}
 q_{\bA}(\bu_B,\bu_I;\om) = \bu_B^T (\bK_{BB} -\om^2 \bM_{BB} ) \bu_B +
 2\bu_B^T \bK_{BI} \bu_I + \bu_I^T (\bK_{II} - \om^2 \bM_{II}) \bu_I.
 \label{eq:quaddyn}
\end{equation}

\begin{remark}
Unlike in the static case the quadratic form $q_{\bA}(\bu_B,\bu_I;\om)$  could
be unbounded from below for $\bu_B$ fixed. This happens for example if there
is a $\bu_I$ so that $\bK_{II} \bu_I \neq \bzero$ and $\bM_{II} \bu_I \neq
\bzero$. Then for $\om$ large enough the matrix in the last term of $q_\bA$
becomes indefinite.  Thus we cannot define the response function at the
terminals through a minimization principle similar to \eqref{eq:quad}.
\end{remark}

The dynamic response function at the terminals is the displacement-to-forces map
at the critical point $\nabla_{\bu_I} q_\bA(\bu_B,\bu_I;\om) = \bzero$, if such
critical point exists. Because $q_\bA$ may be unbounded below, this
critical point could be a saddle point for the quadratic $q_\bA$ with $\bu_B$
fixed. The frequencies $\om$ for which there is no critical point (i.e. there
is some $\bu_B$ so that $\nabla_{\bu_I} q_\bA(\bu_B,\bu_I;\om) \neq \bzero$ for
all $\bu_I$) are important physically and in our derivation and correspond to
the {\em resonance frequencies} of the network. 

To give an expression for the dynamic response function we partition the
interior nodes into nodes $J$ with positive mass and massless nodes $L$, so
that $I = J \cup L$. Therefore $\bM_{JJ}$ is positive definite but $\bM_{LL} =
\bzero$.

\begin{lemma}
 \label{lem:schurdyn}
 The response function at the terminals is
 \begin{equation}
   \bW(\om) = \btK_{BB} - \om^2 \bM_{BB}  - \btK_{BJ} (\btK_{JJ} - \om^2
   \bM_{JJ})^{-1} \btK_{JB},
  \label{eq:schurdyn}
 \end{equation}
 provided that $\om^2$ is not an eigenvalue of $\bM_{JJ}^{-1/2} \btK_{JJ}
 \bM_{JJ}^{-1/2}$. Here we have used the submatrices of the matrix
 \begin{equation}
  \btK
  =
  \begin{bmatrix}
   \btK_{BB} & \btK_{BJ}\\
   \btK_{JB} & \btK_{JJ}
  \end{bmatrix}
  = 
  \begin{bmatrix}
   \bK_{BB} & \bK_{BJ}\\
   \bK_{JB} & \bK_{JJ}
  \end{bmatrix}
  -
  \begin{bmatrix}
   \bK_{BL}\\ \bK_{JL}
  \end{bmatrix}
  \bK_{LL}^\dagger
  \begin{bmatrix}
  \bK_{LB} & \bK_{LJ}
  \end{bmatrix}.
  \label{eq:tk}
 \end{equation}
\end{lemma}
\begin{proof}
 The matrix $\btK$ is the response matrix for the network with
 terminals $B \cup J$ and interior nodes $L$ and can be obtained from
 \eqref{eq:schur}. Since the nodes $L$ are massless the dynamic response at
 the nodes $B \cup J$ is $\btK - \om^2 \diag(\bM_{BB},\bM_{JJ})$. Since
 $\bM_{JJ}$ is non-singular, the matrix $\btK_{JJ} - \om^2 \bM_{JJ}$ is
 singular if and only if $\om^2$ is an eigenvalue of $\bM_{JJ}^{-1/2}
 \btK_{JJ} \bM_{JJ}^{-1/2}$. Thus when $\om^2$ is not an eigenvalue of
 $\bM_{JJ}^{-1/2} \btK_{JJ} \bM_{JJ}^{-1/2}$, we can equilibrate forces at
 the nodes $J$ and get the expression for the response function. 
\end{proof}

A corollary of \lemref{lem:schurdyn} is that if $\om$ is a resonance
frequency of the network then $\om^2$ must be an eigenvalue of the
matrix $\bM_{JJ}^{-1/2} \btK_{JJ} \bM_{JJ}^{-1/2}$.
The expression for the response function in \lemref{lem:schurdyn} leads to the
following properties.
\begin{lemma}
\label{lem:dynamic}
The response function $\bW(\om)$ of any network of springs and
masses with $n$ terminals is of the form
\begin{equation}
 \bW(\om) = \bA - \om^2 \bM + \sum_{i=1}^p \frac{\bC^{(i)}}{\om^2 -
 \om_i^2} \; \in \; \real^{nd \times nd},
 \label{eq:wom}
\end{equation}
where the matrix $\bM = \diag(m_1\be,\ldots,m_n\be)$ is real diagonal with
the masses of the boundary nodes in the diagonal, the vector $\be =
[1,\ldots,1]^T \in \real^d$, the matrices $\bC^{(i)}$ are real symmetric
positive semidefinite, and the static response 
\[
 \bW(0) = \bA - \sum_{i=1}^p \om_i^{-2} \bC^{(i)},
\]
is real symmetric positive semidefinite and balanced (i.e. it satisfies the
conditions (a)--(d) of \lemref{lem:static}). The resonant frequencies are distinct, finite and satisfy $\om_i^2 > 0$.
\end{lemma}
\begin{proof}
Let $I = J \cup L$ be a partition of the interior nodes into massless nodes
$L$ and nodes with positive mass $J$, and let $\btK$ be defined as in
\eqref{eq:tk}. By \lemref{lem:schurdyn} the response function at the terminals
can be rewritten as
\[
 \bW(\om) = \btK_{BB} - \om^2 \bM_{BB} - \btK_{BJ} \bM_{JJ}^{-1/2} (\bC - \om^2 \bI)^{-1}  \bM_{JJ}^{-1/2} \btK_{JB},
\]
where $\bC = \bM_{JJ}^{-1/2} \btK_{JJ} \bM_{JJ}^{-1/2}$ and $\om^2$ is not an
eigenvalue of $\bC$. The matrix $\bC$ is symmetric positive semidefinite
because $\btK$ is a symmetric positive semidefinite matrix
(\lemref{lem:static}).  Let $\{ \om_j^2 \}_{j=1}^N$ be the (nonnegative)
eigenvalues of $\bC$ and $\{ \bc_j \}_{j=1}^N$ be a corresponding orthonormal
basis of eigenvectors of $\bC$, where $N=|J|$. When $\om \neq 0$ and $\om^2
\neq \om_j^2$, the response function $\bW(\om)$ becomes
\begin{equation}
 \label{eq:swom}
 \bW(\om) = \btK_{BB} - \om^2 \bM_{BB} + \sum_{j=1}^N \frac{\btc_j
 \btc_j^T}{\om^2 - \om_j^2},
\end{equation}
with $\btc_j = \btK_{BJ} \bM_{JJ}^{-1/2} \bc_j$, $j=1,\ldots,N$.
Let $r = \rank(\bC) = \rank(\btK_{JJ})$ and assume the eigenvalues are
ordered such that $\om_j^2>0$ for $j=1,\ldots,r$. Clearly $\btK_{JJ}
\bM_{JJ}^{-1/2} \bz = \bzero$ if and only if $\bC\bz = \bzero.$ Thus
$\bM_{JJ}^{-1/2} \bc_j \in \nullspace(\btK_{JJ})$, for $j = r+1,\ldots, N$.
By \lemref{lem:floppy}, we have $\nullspace(\btK_{JJ}) \subset
\nullspace(\btK_{JB})$ which means that 
\[
 \btc_j = \bzero ~\text{for}~ j=r+1,\ldots,N.
\]
In other words, only the first $r$ terms of the sum in
\eqref{eq:swom} are nonzero. We obtain the form \eqref{eq:wom} of the
response function from \eqref{eq:swom} by setting $\bA = \btK_{BB}$ and $\bM =
\bM_{BB}$. The matrices $\bC^{(i)}$ are the sum of the matrices $\btc_j
\btc_j^T$ that correspond to the same resonance $\om_i^2$, thus the
$\bC^{(i)}$ must be real symmetric positive semidefinite.

We now show that $\bW(0) = \btK_{BB} - \btK_{BJ} \btK_{JJ}^\dagger
\btK_{JB}$, i.e. at $\om=0$ the dynamic response function is the static
response of the network. Then the properties of $\bW(0)$ follow from
\lemref{lem:static}.  First note that from \eqref{eq:swom},
\[
 \bW(0) = \btK_{BB} - \btK_{BJ} \bM_{JJ}^{-1/2} \bC^\dagger \bM_{JJ}^{-1/2} \btK_{JB}
\]
where we used that $$\bC^\dagger = \sum_{j=1}^r \om_j^{-2} \bc_j \bc_j^T.$$ It
is sufficient to show that $\bu_J = -  \bM_{JJ}^{-1/2} \bC^\dagger
\bM_{JJ}^{-1/2} \btK_{JB} \bu_B$ equilibrates the forces at the interior
nodes for any terminal displacements $\bu_B$. Indeed we have 
\begin{equation}
 \label{eq:cpinv}
 \bC \bM_{JJ}^{1/2} \bu_J = (\bM_{JJ}^{-1/2} \btK_{JJ} \bM_{JJ}^{-1/2}) \bM_{JJ}^{1/2} \bu_J = -\bM_{JJ}^{-1/2} \btK_{JB} \bu_B,
\end{equation}
since \lemref{lem:floppy} and $\bM_{JJ}$ invertible imply
$$\range(\bC) = \range(\bM_{JJ}^{-1/2} \btK_{JJ}) \supset
\range(\bM_{JJ}^{-1/2} \btK_{JB}).$$ The balance of forces at the nodes $J$
(i.e. $\btK_{JJ} \bu_J = - \btK_{JB} \bu_B$) comes from multiplying
\eqref{eq:cpinv} by $\bM_{JJ}^{1/2}$ on the left.  
\end{proof}

\subsection{Network transformations not affecting the response function}
\label{sec:transf}
We need a few elementary transformations that allow more flexibility
with the placement of springs in a network. We assume
throughout this text that the springs occupy an arbitrarily small volume
and that the nodes are points which may or may not have a mass attached
to them.

\subsubsection{Avoiding a line} It is possible to transform a spring in
order to avoid a line or a crossing (for networks in $\real^3$). The
construction for networks $\real^3$ is given in \cite[Example
3.2]{milton_seppecher2008} and consists of replacing the spring by a {\em
simple truss} as shown in \figref{fig:truss}(a). A similar construction
can be done for planar networks, see \figref{fig:truss}(b). If a network
in $\real^3$ has springs crossing, these can be eliminated via this
transformation since the position of the additional interior nodes $\bu,
\bv,\bw$ is not fixed. Moreover, the additional nodes in the truss structures can be
chosen to avoid a finite number of points.

\begin{figure}
 \begin{center}
  \begin{tabular}{c@{\hspace{1cm}}c}
\begin{picture}(0,0)%
\includegraphics{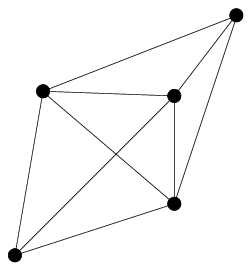}%
\end{picture}%
\setlength{\unitlength}{1184sp}%
\begingroup\makeatletter\ifx\SetFigFont\undefined%
\gdef\SetFigFont#1#2#3#4#5{%
  \reset@font\fontsize{#1}{#2pt}%
  \fontfamily{#3}\fontseries{#4}\fontshape{#5}%
  \selectfont}%
\fi\endgroup%
\begin{picture}(3930,4723)(1861,-5999)
\put(1876,-2986){\makebox(0,0)[lb]{\smash{{\SetFigFont{10}{12.0}{\rmdefault}{\mddefault}{\updefault}{\color[rgb]{0,0,0}$\bv$}%
}}}}
\put(5776,-1711){\makebox(0,0)[lb]{\smash{{\SetFigFont{10}{12.0}{\rmdefault}{\mddefault}{\updefault}{\color[rgb]{0,0,0}$\bx_j$}%
}}}}
\put(4876,-4861){\makebox(0,0)[lb]{\smash{{\SetFigFont{10}{12.0}{\rmdefault}{\mddefault}{\updefault}{\color[rgb]{0,0,0}$\bu$}%
}}}}
\put(2176,-5836){\makebox(0,0)[lb]{\smash{{\SetFigFont{10}{12.0}{\rmdefault}{\mddefault}{\updefault}{\color[rgb]{0,0,0}$\bx_i$}%
}}}}
\put(4726,-3136){\makebox(0,0)[lb]{\smash{{\SetFigFont{10}{12.0}{\rmdefault}{\mddefault}{\updefault}{\color[rgb]{0,0,0}$\bw$}%
}}}}
\end{picture}%
&%
\begin{picture}(0,0)%
\includegraphics{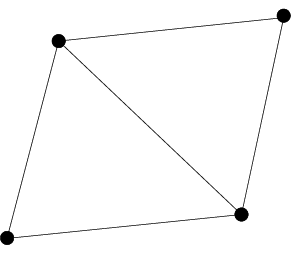}%
\end{picture}%
\setlength{\unitlength}{1184sp}%
\begingroup\makeatletter\ifx\SetFigFont\undefined%
\gdef\SetFigFont#1#2#3#4#5{%
  \reset@font\fontsize{#1}{#2pt}%
  \fontfamily{#3}\fontseries{#4}\fontshape{#5}%
  \selectfont}%
\fi\endgroup%
\begin{picture}(4704,4573)(2062,-5699)
\put(6076,-4786){\makebox(0,0)[lb]{\smash{{\SetFigFont{10}{12.0}{\rmdefault}{\mddefault}{\updefault}{\color[rgb]{0,0,0}$\bu$}%
}}}}
\put(2251,-5536){\makebox(0,0)[lb]{\smash{{\SetFigFont{10}{12.0}{\rmdefault}{\mddefault}{\updefault}{\color[rgb]{0,0,0}$\bx_i$}%
}}}}
\put(6751,-1561){\makebox(0,0)[lb]{\smash{{\SetFigFont{10}{12.0}{\rmdefault}{\mddefault}{\updefault}{\color[rgb]{0,0,0}$\bx_j$}%
}}}}
\put(2776,-1636){\makebox(0,0)[lb]{\smash{{\SetFigFont{10}{12.0}{\rmdefault}{\mddefault}{\updefault}{\color[rgb]{0,0,0}$\bv$}%
}}}}
\end{picture}%
\\
  (a) & (b)
  \end{tabular}
 \end{center}
 \caption{Truss structure replacing a spring between nodes $\bx_i$ and
 $\bx_j$ without changing the response function. This structure can be
 used to avoid a line for (a) networks in $\real^3$ or (b) planar
 networks.}
 \label{fig:truss}
\end{figure}

\subsubsection{Virtual crossings} A network with all springs in
$\real^2$ is not necessarily planar because its springs may cross.
However \cite[Example 3.15]{milton_seppecher2008} shows how to replace
such a crossing by a planar network with exactly the same response
function. This transformation involves adding a node at the crossing point
of the springs and carefully choosing the spring constants. To avoid a
finite number of points one can first replace one of the springs by a
simple truss as in \figref{fig:truss}(b) and use virtual crossings to
transform the network into a planar network.

\subsection{The superposition principle} A fundamental tool for our
construction of a network reproducing the response function is the
following result, which is valid for both planar and $\real^3$ networks.

\begin{lemma}
Let $\bW_1$ and $\bW_2$ be the response matrices of two networks (planar
or in $\real^3$, static or dynamic) sharing the same terminals but with
no interior nodes in common. Then the response function of both networks
together is $\bW_1 + \bW_2$.
\end{lemma}
\begin{proof}
The result follows from the reasoning in \cite[Remark
3.9]{milton_seppecher2008} and the frequency independent transformations
in \secref{sec:transf}. For the planar case any crossing can be
eliminated using \cite[Example 3.15]{milton_seppecher2008}. Note that
the transformations in \secref{sec:transf} allow  one to
avoid a finite number of locations (except the terminals).
\end{proof}

\section{Characterization of the static response}
\label{sec:static}

 Building upon the seminal work of Camar-Eddine and Seppecher
 \cite{camareddine_seppecher2003}, we give necessary and sufficient
 conditions for a function to be the response function for either a planar
 network or a network in $\real^3$.  The sufficiency is proved
 constructively and relies on the existence of networks that have rank
 one response matrices, as is described in detail in the remaining part
 of \secref{sec:static}.

 Recall that the $\epsilon$-neighborhood $\bC_\epsilon$ of
a set $\bC$ is the set,
\begin{equation}
\label{4.0}
\bC_\epsilon= \{\bx \in \real^2 ~|~ \mathrm{dist}(\bx, \bC)\leq \epsilon\}.
\end{equation}
If the set $\bC$ is convex then the set $\bC_\epsilon$ is also convex
because of the convexity of the function $\mathrm{dist}(\bx,
\bC)$ for convex $\bC$ (see e.g. \cite[\S 3.2.5]{boyd_vandenberghe2004}).

 \begin{theorem}
 \label{thm:charstat}
 For any choice of terminal node positions, any matrix $\bW$ satisfying
 the properties in \lemref{lem:static} is the response matrix of a
 purely elastic network which is either planar or in $\real^3$.
 Moreover, any internal nodes in the construction can be chosen within
 an $\eps-$neighborhood of the convex hull of the terminals, and
 avoiding a finite number of positions.
 \end{theorem}
 \begin{proof}
 By properties (a)--(c) in \lemref{lem:static} the matrix $\bW$ can
 be written as a sum of rank one matrices $\bW_i$:
 \[
  \bW = \sum_{i=1}^n \bW_i,
 \]
 where $\bW_i = \lambda_i \bw_i \bw_i^T$ and $(\lambda_i,\bw_i)$ is an
 eigenpair of $\bW$, $\lambda_i\geq 0$,
 for $i=1,\ldots,n$.  Each $\bW_i$ satisfies properties (a)--(d) in
 \lemref{lem:static}. Properties (a)--(c) are easy to check for $\bW$ and (d)
 follows by linearity, since it holds for each $\bW_i$. Owing to
 \thmref{thm:rankone_gen} (\thmref{thm:rankone} in the planar case), it
 is possible to construct a network with matrix response equal to
 $\bW_i$. By the superposition principle we obtain a general network
 with response $\bW$. If the desired network is planar, then every
 crossing between springs can be transformed to a planar network through
 a truss-like structure \cite[Examples 3.2, 3.15]{milton_seppecher2008}.
 As discussed in \secref{sec:transf}, such transformations can be chosen
 to avoid a finite number of points in $\real^d$.
 \end{proof}

 \begin{remark}
  \label{rem:rescale}
  If $\bW$ is the static response matrix of a network and $\alpha$ is a
  positive constant then clearly $\alpha \bW$ is the response matrix of the
  same spring network, but where all the spring constants are multiplied by
  $\alpha$.
 \end{remark}

\subsection{Planar networks with rank one static response matrices}
The main result in this section is \thmref{thm:rankone} which is the
statement of \thmref{thm:charstat} for rank one response matrices, i.e.
it shows that for any rank one response matrix satisfying
\lemref{lem:static} it is possible to find a planar network that
realizes it. We first prove \thmref{thm:rankone} for three terminal
networks in \secref{sec:threeterm}, then for four and more terminals in
\secref{sec:fourterm}.
\begin{figure}
 \begin{center}
\begin{picture}(0,0)%
\includegraphics{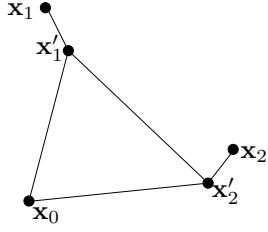}%
\end{picture}%
\setlength{\unitlength}{1184sp}%
\begingroup\makeatletter\ifx\SetFigFont\undefined%
\gdef\SetFigFont#1#2#3#4#5{%
  \reset@font\fontsize{#1}{#2pt}%
  \fontfamily{#3}\fontseries{#4}\fontshape{#5}%
  \selectfont}%
\fi\endgroup%
\begin{picture}(4905,4849)(1711,-5549)
\put(1726,-1111){\makebox(0,0)[lb]{\smash{{\SetFigFont{10}{12.0}{\rmdefault}{\mddefault}{\updefault}{\color[rgb]{0,0,0}$\bx_1$}%
}}}}
\put(2326,-1936){\makebox(0,0)[lb]{\smash{{\SetFigFont{10}{12.0}{\rmdefault}{\mddefault}{\updefault}{\color[rgb]{0,0,0}$\bx_1'$}%
}}}}
\put(2251,-5386){\makebox(0,0)[lb]{\smash{{\SetFigFont{10}{12.0}{\rmdefault}{\mddefault}{\updefault}{\color[rgb]{0,0,0}$\bx_0$}%
}}}}
\put(6001,-4936){\makebox(0,0)[lb]{\smash{{\SetFigFont{10}{12.0}{\rmdefault}{\mddefault}{\updefault}{\color[rgb]{0,0,0}$\bx_2'$}%
}}}}
\put(6601,-4111){\makebox(0,0)[lb]{\smash{{\SetFigFont{10}{12.0}{\rmdefault}{\mddefault}{\updefault}{\color[rgb]{0,0,0}$\bx_2$}%
}}}}
\end{picture}%
 \end{center}
 \caption{A planar network with rank one response. Any
 resulting force at terminals is proportional to $(\bbf_0^T,
 \bbf_1^T, \bbf_2^T)^T$.}
 \label{fig:rankone}
\end{figure}

\subsubsection{Three terminal rank one static planar networks}
\label{sec:threeterm} We show how to construct a three terminal planar
elastic network realizing any valid rank one response matrix. If
$\bbf_0,\bbf_1,\bbf_2 \in \real^2$ is the balanced system of forces at
the nodes $\bx_0, \bx_1, \bx_2$, the construction depends on $\rank
[\bbf_1, \bbf_2, \bx_1 - \bx_0, \bx_2 - \bx_0]$. More precisely
\lemref{lem:threeterm} corresponds to the case when this rank is two and
\lemref{lem:threeterm-1D} when this rank is one. Since we are in
$\real^2$ these are  the only non-trivial cases available, which shows
\thmref{thm:rankone} for planar three terminal networks.

\begin{remark} (Pierre Seppecher, private communication)
The easiest way to construct a three terminal rank one
network is to add a node at the intersection of the force lines
(three forces that are balanced meet at a single point in 2D, this
can be shown by writing the torque balance equation for the
intersection point). The only problem with this construction is
that the extra node can end up far away if the force lines are almost
parallel.
\end{remark}

We start with the following intermediate result. A similar result is shown in
three dimensions by Camar-Eddine and Seppecher \cite[Lemma
5]{camareddine_seppecher2003}.
\begin{lemma} 
\label{lem:notcolli}
Let $\bbf_0$, $\bbf_1$, $\bbf_2$  be a set of balanced forces at the
nodes $\bx_0$, $\bx_1$, $\bx_2$ in $\real^2$. Then if $\rank [\bbf_1,
\bbf_2, \bx_1 - \bx_0, \bx_2 - \bx_0] = 2$, there is an $\eps > 0$ such
that the points $\bx_0$, $\bx_1'=\bx_1+\eps \bbf_1$ and $\bx_2' = \bx_2
+ \eps \bbf_2$ are not collinear. Moreover $\eps$ can be chosen
arbitrarily small and so that $\bx_1'$ and $\bx_2'$ do not coincide with
a finite number of points. 
\end{lemma}
\begin{proof}
If it were true that for all $\eps  > 0$ the three points $\bx_0$,
$\bx_1'$ and $\bx_2'$ are collinear, then the second degree polynomial
in $\eps$, $p(\eps) = \det(\bx_1' - \bx_0, \bx_2' - \bx_0)$ is
identically zero. That the constant coefficient of $p(\eps)$ vanishes
means that $\det (\bx_1-\bx_0,\bx_2 - \bx_0)=0$, or that the points
$\bx_0$, $\bx_1$ and $\bx_2$ are collinear. Since the lemma for
$\bx_0=\bx_1=\bx_2$ is trivial to prove, we may assume without loss of
generality that there is some $\alpha \in \real$ such that
\begin{equation}
 \bx_2 - \bx_0 = \alpha (\bx_1 - \bx_0),
 \label{eq:notcolli1}
\end{equation}
swapping the indices $1$ and $2$ if necessary.
Since the coefficient in $\eps$ of $p(\eps)$ vanishes we get
$\det(\bbf_1, \bx_2 - \bx_0) + \det(\bx_1 - \bx_0, \bbf_2)
= 0$ or equivalently $\det(\bx_1 - \bx_0, - \alpha \bbf_1 +
\bbf_2) = 0$. Now the torque balance implies that
\begin{equation}
 (\bx_1 - \bx_0) \cross (\bbf_1 + \alpha \bbf_2) = 0
 \label{eq:notcolli2}
\end{equation}
Putting both \eqref{eq:notcolli1} and \eqref{eq:notcolli2} in
matrix form, there are some real $\beta$ and $\gamma$ such that
\[
 \begin{bmatrix} \bI & \alpha \bI \\-\alpha \bI & \bI
 \end{bmatrix} \begin{bmatrix} \bbf_1\\ \bbf_2 \end{bmatrix} =
 \begin{bmatrix} \beta (\bx_1 - \bx_0) \\ \gamma (\bx_1 -
 \bx_0) \end{bmatrix}.
\]
The determinant of the matrix above is $(\alpha^2 + 1)^2 \neq
0$, thus $\rank [\bbf_1, \bbf_2, \bx_1 - \bx_0, \bx_2 - \bx_0]
= 1$ which contradicts the hypothesis of the lemma. Finally since
$p(\eps)$ is not identically zero, one can choose an arbitrarily small
$\eps$ that avoids a finite number of points.
\end{proof}

\begin{lemma} \label{lem:threeterm} Let $\bbf_0$, $\bbf_1$, $\bbf_2$  be
a set of balanced forces at the nodes $\bx_0$, $\bx_1$, $\bx_2$ in
$\real^2$. If $\rank [\bbf_1, \bbf_2, \bx_1 - \bx_0, \bx_2 -
\bx_0] = 2$ then there exists a purely elastic planar network with
force response proportional to $(\bbf_0^T, \bbf_1^T, \bbf_2^T)^T$. The
internal nodes of such a network can be chosen to avoid a finite number
of points and within an $\eps-$neighborhood of the convex hull of the terminals.
\end{lemma}
\begin{proof}
First observe the hypothesis that $\rank [\bbf_1, \bbf_2, \bx_1 -
\bx_0, \bx_2 - \bx_0] = 2$ implies the existence of at least one
permutation $\sigma$ of $\{0,1,2\}$ such that
\begin{equation}
\label{stea}\bbf_{\sigma(1)}\cross(\bx_{\sigma(1)}-\bx_{\sigma(0)})\neq
0
\end{equation}
Indeed, if \eqref{stea} is false for all permutations $\sigma$,
i.e.,
\[
\begin{aligned}
\label{stea'} 
\bbf_1\cross(\bx_1-\bx_0) & = 0 =\bbf_2\cross(\bx_2-\bx_0),\\
\bbf_2\cross(\bx_2-\bx_1) & = 0 =\bbf_0\cross(\bx_0-\bx_1),\\
\bbf_0\cross(\bx_0-\bx_2) & = 0 =\bbf_1\cross(\bx_1-\bx_2)
\end{aligned}
\] we have that
$\rank [\bbf_1, \bbf_2, \bx_1 - \bx_0, \bx_2 - \bx_0] = 1$ and
this contradicts our initial hypothesis. So, without loss of 
generality we can assume that,
\begin{equation}
\label{stea''} \bbf_1\cross(\bx_1-\bx_0)\neq 0.\end{equation}

 Next, let $\bx_i' = \bx_i + \eps \bbf_i$, for $i=1,2$. By
\lemref{lem:notcolli}, if $\rank [\bbf_1, \bbf_2, \bx_1 - \bx_0,
\bx_2 - \bx_0] = 2$, there is $\eps > 0$ such that the points
$\bx_0$, $\bx_1'$ and $\bx_2'$ are not collinear and do not coincide
with a finite number of points. Consider the
network in \figref{fig:rankone} (the spring constants are
irrelevant for this proof). Let $\bA$ be the response matrix for
the network including both terminal ($B=\{0,1,2\}$) and interior
($I=\{1',2'\}$) nodes. The associated quadratic form is
\begin{equation}
 q_{\bA}(\bu_B,\bu_I) = \sum_{i=1}^2 s_{(\bx_0,\bx_i')}(\bu_0,\bu_i') +
 \sum_{i=1}^2 s_{(\bx_i,\bx_i')}(\bu_i,\bu_i') +
 s_{(\bx_1',\bx_2')}(\bu_1',\bu_2').
 \label{eq:qa}
\end{equation}
It suffices to show that the quadratic form $q_{\bM}(\bu_B)$
at the terminal nodes has codimension one (since $\bM$ is
positive semidefinite, we do have $\ker \bM = \ker q_{\bM}$).
Actually $\bu_B = (\bu_0^T,\bu_1^T,\bu_2^T)^T \in \ker q_{\bM}$ if
and only if there exists $\bu_I = (\bu_1'^T,\bu_2'^T)^T$ such that
all terms in the sum \eqref{eq:qa} vanish or equivalently,
\begin{align}
  (\bu_i'-\bu_0)^T  (\bx_i'- \bx_0) &= 0,~\text{for
  $i=1,2$,}\label{eq:propa}\\
  (\bu_i'-\bu_i)^T  \bbf_i &=0,~\text{for
  $i=1,2$,}\label{eq:propb}\\
  (\bu_2'-\bu_1')^T  (\bx_2' - \bx_1') &=0.\label{eq:propm}
\end{align}
Property \eqref{eq:propa} is equivalent to
\[
 \begin{aligned}
  \bu_1' - \bu_0 = a \bR_\perp (\bx_1' - \bx_0)\\
  \bu_2' - \bu_0 = b \bR_\perp (\bx_2' - \bx_0)\\
 \end{aligned}
\]
for some $a,b$ reals and where
\[
 \bR_\perp = \begin{bmatrix} 0 & 1\\ -1 & 0 \end{bmatrix}.
\]
Property \eqref{eq:propm} implies that $b=a$, i.e. the only
infinitesimal deformation which does not change the side lengths of the
triangle in \figref{fig:rankone} is an infinitesimal rigid motion (translation
plus rotation).  Since $\bbf_i \cross \bx_i' = \bbf_i \cross \bx_i$ and
the forces and torques are balanced we conclude from \eqref{eq:propb} that
\[
\sum_{i=1}^2 \bbf_i \cdot (\bu_i - \bu_0) =  a \sum_{i=1}^2 \bbf_i
 \cross (\bx_i' - \bx_0) = a \sum_{i=0}^2 \bbf_i \cross
 \bx_i = 0.
\]
Thus $\ker q_\bM \subset
\linspan\{(\bbf_0^T,\bbf_1^T,\bbf_2^T)^T\}^\perp$.\\
To prove the other side of the inclusion we start with $\bu_i$,
$i=0,1,2$ satisfying
\begin{equation}
  \sum_{i=1}^2 \bbf_i \cdot (\bu_i - \bu_0) = 0
 \label{eq:orthoforce}
\end{equation}
and we seek $\bu_i'$, $i=1,2$ such that \eqref{eq:propa},
\eqref{eq:propb} and \eqref{eq:propm} hold. For any real $a$, the choice
\[
 \bu_1' = \bu_0 + a \bR_\perp (\bx_1'-\bx_0)
 ~~\text{and}~~
 \bu_2' = \bu_0 + a \bR_\perp (\bx_2'-\bx_0),
\]
satisfies \eqref{eq:propa} and \eqref{eq:propm}.
Next, note
that from the definition of the points $\bx_1',\bx_2'$ we have,
$$(\bx_i'-\bx_0) \cross \bbf_i = (\bx_i-\bx_0) \cross \bbf_i,\mbox{ for }i=1,2.$$
Using the latter, the balance of forces and \eqref{stea''}, property
\eqref{eq:propb} follows by taking
\[
a = \frac{\bbf_1\cdot (\bu_1 - \bu_0)}{\bbf_1 \cross (\bx_1 -
\bx_0) }
=   \frac{\bbf_2 \cdot (\bu_2 - \bu_0)}{\bbf_2 \cross (\bx_2 -
\bx_0)}.
\]
This proves that
$\ker q_\bM = \linspan\{(\bbf_0^T,\bbf_1^T,\bbf_2^T)^T\}^\perp$, so
$q_\bM$ can be written for some $c>0$ as,
\[
 q_\bM(\bu_B) = c \M{\sum_{i=0}^2 \bbf_i \cdot \bu_i}^2.
\]
\end{proof}

\begin{lemma}
\label{lem:threeterm-1D}Let $\bbf_0$, $\bbf_1$, $\bbf_2$  be a set
of balanced forces at the nodes $\bx_0$, $\bx_1$, $\bx_2$ in
$\real^2$. If $\rank [\bbf_1, \bbf_2, \bx_1 - \bx_0, \bx_2 -
\bx_0] = 1$ then there exists a purely elastic planar network with
force response proportional to $(\bbf_0^T, \bbf_1^T, \bbf_2^T)^T$. The
internal nodes of such a network can be chosen to avoid a finite number
of points and within an $\eps-$neighborhood of the convex hull of the terminals.
\end{lemma}

\begin{proof} We build up on the idea presented in \cite[Theorem
5]{camareddine_seppecher2003}. It can be easily observed (by choosing a
coordinate system with $\bx_0$ as origin) that for any point $\by$ in
the plane, not collinear with $\bx_0, \bx_1, \bx_2$, there is a force
$\bbf$, such that the following families,
\begin{equation}
\label{2stea'}
 ((\bbf,\by), (\bbf_0+\bbf_2-\bbf,\bx_0), (\bbf_1,\bx_1))
 ~\text{and}~
((-\bbf,\by), (\bbf-\bbf_2,\bx_0), (\bbf_2,\bx_2)),
\end{equation} 
form balanced systems of
forces. Then, by using \lemref{lem:threeterm} there exists
purely elastic networks with response matrices proportional with
$(\bbf^T, \bbf_0^T+\bbf_2^T - \bbf^T, \bbf_1^T)^T$, and $(-\bbf^T,
\bbf^T-\bbf_2^T, \bbf_2^T)^T$ respectively. It is possible to choose the
spring constants so that the associated quadratic forms are (see
\remref{rem:rescale}), 
\begin{equation} 
 \label{2stea}
 \begin{aligned}
  q'(\bv,\bu_0,\bu_1)&= (\bbf\cdot \bv + (\bbf_2-\bbf)\cdot \bu_0+\bbf_0\cdot
  \bu_0 + \bbf_1\cdot\bu_1)^2 ~\text{and} \\
  q''(\bv,\bu_0,\bu_2)&= (-\bbf \cdot \bv +(\bbf -\bbf_2) \cdot \bu_0 +
  \bbf_2\cdot\bu_2)^2.
\end{aligned}
\end{equation} 
Now, let us consider the infimum,
\begin{equation}
 \label{2stea''}
 \tq(\bu_0, \bu_1, \bu_2)=\inf_{\bv\in{\real^2}} \{ q'(\bv,\bu_0,\bu_1)
 + q''(\bv,\bu_0,\bu_2) \}.
\end{equation}
The necessary condition of the infimum,
\begin{equation}
\label{2stea'''} 
 \bbf\cdot\bv=- \frac{1}{2} \M{(\bbf_2-\bbf)\cdot\bu_0+\bbf_0\cdot\bu_0+\bbf_1\cdot\bu_1}- \frac{1}{2}\M{(\bbf_2 - \bbf)\cdot\bu_0-\bbf_2\cdot\bu_2},
\end{equation}
implies the statement of the Lemma, i.e.,
\begin{equation}
\label{2stea''''} \tq(\bu_0, \bu_1,
\bu_2)=\frac{1}{2} \left(\sum_{i=0}^{2}\bbf_i\cdot\bu_i\right)^2.\end{equation} The point $\by$ and
any additional points in \lemref{lem:threeterm} can be chosen to avoid a
finite number of points in the plane, and within an $\eps-$neighborhood
of the convex hull of the terminals.
\end{proof}

\subsubsection{General rank one planar networks}
\label{sec:fourterm}
We show in \lemref{lem:fourterm} the construction of a network realizing
a valid four terminal rank one response and then generalize the result to
any number of terminals in \thmref{thm:rankone}.  We start our argument
with \lemref{lem:fourterm'} which is a technical result needed later in
this section.
\begin{lemma}
\label{lem:fourterm'} Let $\bbf_0, \bbf_1, \bbf_2, \bbf_3$ and
$\{\bx_0,\bx_1,\bx_2,\bx_3\}$ be a balanced system of forces in
$\real^2$. Then there exists a point $\by_*$ in an
$\epsilon$-neighborhood of the convex hull of the set $\{\bx_0, \bx_1,
\bx_2, \bx_3\}$ such that $\by_*\notin \{\bx_0,\bx_1,\bx_2,\bx_2\}$ and
\[
 \bbf_i\cross \bx_i+\bbf_j\cross
\bx_j-(\bbf_i+\bbf_j)\cross \by_*=0, \mbox{ for some }
\{i,j\}\in\{0,1,2,3\},~i\neq j,
\]
or in other words, the three forces $\bbf_i$, $\bbf_j$ and $-(\bbf_i+\bbf_j)$
supported at the nodes $\bx_i$, $\bx_j$ and $\by_*$ are balanced. The point
$\by_*$ can be chosen to avoid a finite number of positions.
\end{lemma}

\begin{proof}
For any pair of indices $\{i,j\}$, with $i\neq j$ and $i,j \in
\{0,1,2,3\}$, let $\bbf_{ij}:\real^2\rightarrow \real$ be defined by
\begin{equation}
 \label{4.1}
\bbf_{ij}(\by)=\bbf_i\cross \bx_i+\bbf_j\cross \bx_j-(\bbf_i+\bbf_j)\cross
\by.
\end{equation}
Using the balance of forces relations it can be easily observed
that
\begin{equation}
\label{4.2} \bbf_{ij}=\bbf_{ji}=-\bbf_{kt}=-\bbf_{tk},\mbox { for
any } \{i,j,k,t\}=\{0,1,2,3\}.
\end{equation}
Let $\bC$ be the convex hull of $\{\bx_0,\bx_1,\bx_2,\bx_3\}$, that is
 \[
  \bC=\Mcb{\bz \in \real^2 ~ \big| ~ \bz=\sum_{i=0}^3c_i\bx_i, \mbox{ for } c_i \geq 0
  \mbox { and } \sum_{i=0}^3 c_i=1},
 \]
 and $\bC_\epsilon$ be the $\epsilon$-neighborhood of the set $\bC$ as
 defined in \eqref{4.0}.

 Next we show that there exists a pair of indices $\{i,j\}$ with $i,j\in\{0,1,2,3\}$,
 with the property that there exists a point $\by_*\in \bC_\epsilon$, such that
 $\bbf_{ij}(\by_*)=0$.
 We reason by contradiction. Assume that the above is not true,
 i.e.,

 For all pairs $\{i,j\}$, with $i\neq j$ and $i,j\in\{0,1,2,3\}$
 we have
\begin{equation}
 \label{4.3}
 \bbf_{ij}(\by)\neq 0, \mbox{ for any }
 \by\in \bC_\epsilon.
 \end{equation}

 Using the continuity of the functions $\bbf_{ij}$ and the convexity of
 the set $\bC_\epsilon$, from \eqref{4.3} we obtain that all the
 functions $\bbf_{ij}$ have constant strictly positive or strictly
 negative sign over $\bC_\epsilon$. Using this observation, together
 with the relations \eqref{4.2}, for any partition $
 \{i,j\}\cup\{m,n\}=\{0,1,2,3\}$ we have 
 \begin{equation}
 \label{4.4}
 \bbf_{ij}(\by) \bbf_{mn}(\by) <0 \mbox{ for } \by \in \bC_\epsilon.
\end{equation}
For simplicity we shall call from now the {\em complement} of $\bbf_{ij}$ the function
$\bbf_{kt}$ with $\{k,t\}= \{0,1,2,3\}\setminus \{i,j\}$.

From \eqref{4.2}, \eqref{4.3} and \eqref{4.4} we conclude there
are six different functions $\bbf_{ij}$, with
$\{i,j\}\in\{0,1,2,3\}$, out of which three functions have
strictly positive sign on the set $\bC_\epsilon$ while their
complements have strictly negative sign on $\bC_\epsilon$. Using
this observation, it can be easily checked that for
$\{i,j,k,t\}=\{0,1,2,3\}$, at least one of the following triplets
of functions $(\bbf_{ij},\bbf_{ik},\bbf_{it})$,
 $(\bbf_{ij},\bbf_{ik},\bbf_{jk})$,
 $(\bbf_{ij},\bbf_{it},\bbf_{jt})$,
 $(\bbf_{ij},\bbf_{jt},\bbf_{jk})$,
 $(\bbf_{kt},\bbf_{ik},\bbf_{it})$,
has a constant strict sign over the set $\bC_\epsilon$. Indeed, if $\bbf_{ij}$, $\bbf_{ik}$ and $\bbf_{it}$ do not have the same sign, either two are positive and one is negative or two are negative and one is positive. By replacing one by its complement we get three functions of the same sign.
From the balance of forces relations one can immediately observe
that the sum of the functions in any of these triplets is of the form
\begin{equation}
\label{4.5}
\pm 2\bbf_{i_0}\cross(\by-\bx_{i_0})\mbox{ for some } i_0\in\{0,1,2,3\}.
\end{equation}
Finally relation \eqref{4.5} leads to a contradiction.  Indeed for a
triplet with the property that all the functions in the triplet have the
same strict sign, the sum of its functions must have the same sign on
$\bC_\epsilon$. However from \eqref{4.5} the sum cannot have constant
sign over the set $\bC_\epsilon$ because it equals zero for
$\by=\bx_{i_0}$.  Thus the hypothesis \eqref{4.3} is false and we have
that there exists a pair of distinct indices $\{i,j\} \subset
\{0,1,2,3\}$ so that $\bbf_{ij}(\by_*)=0$ for some $\by_*\in
\bC_\epsilon$.
 
From \eqref{4.5} it is possible to choose $\by_*\notin
\{\bx_0,\bx_1,\bx_2,\bx_2\}$. Indeed from \lemref{lem:fourterm} there
are two indices $\{m,n\}$ such that
 \begin{equation}
 \label{4.7}
 \bbf_{mn}(\by_*)=0.
 \end{equation}
Consider the point $\by_\delta=\by_*+\delta(\bbf_m+\bbf_n)$ with
$\delta>0$. For $\delta$ small enough it is clear that
$\by_\delta\in\bC_\epsilon\setminus\{\bx_0,\bx_1,\bx_2,\bx_3\}$ and
from \eqref{4.1} we obtain
\[
 \begin{aligned}
\bbf_{mn}(\by_\delta)&=\bbf_m\cross
\bx_m+\bbf_n\cross \bx_n-(\bbf_m+\bbf_n)\cross \by_\delta\\
&= \bbf_{mn}(\by_*) -\delta(\bbf_m+\bbf_n)\cross(\bbf_m+\bbf_n)\\
&=  0,
 \end{aligned}
\]
 which implies the statement of the Lemma.
 \end{proof}

The next \lemref{lem:fourterm} shows that for any balanced system of
four forces in $\real^2$ there exists a purely elastic four terminal planar
network with proportional (rank one) force response.  
\begin{lemma}
\label{lem:fourterm} Let $\bbf_0, \bbf_1, \bbf_2, \bbf_3$ and
$\bx_0, \bx_1, \bx_2, \bx_3$ be a balanced system of forces in $\real^2$.
Then, there is a purely elastic four terminal planar network with
force response proportional to $(\bbf_0^T, \bbf_1^T, \bbf_2^T,
\bbf_3^T)^T$. The internal nodes of such a network can be chosen away from a finite number of points, and  within an $\eps-$neighborhood of the convex hull of the terminals. \end{lemma}
\begin{proof}

 From \lemref{lem:fourterm'} we have that there exists a pair of
 indices $\{i,j\}$, a point $\by_*\in\bC_\epsilon\setminus\{\bx_0, \bx_1, \bx_2,\bx_3\}$, and a force
 $\bbf=\bbf_i+\bbf_j$,
  so that the sets $(\by_*,\bbf), (\bx_i,\bbf_i),
 (\bx_j,\bbf_j)$ and $(\by_*,-\bbf), (\bx_{k},\bbf_{k}),
 (\bx_t,\bbf_t)$ are balanced sets of forces.
 From the results of the previous section we have that both sets have a rank one network reproducing the
 forces. By rescaling the spring constants (see \remref{rem:rescale}) their associated quadratic forms are,
 \[
  \begin{aligned}
  q_1(\bv, \bu_i, \bu_j) & = \left(\bbf
  \cdot \bv+\bbf_i\cdot \bu_i + \bbf_j\cdot \bu_j\right)^2\\
  q_2(\bv, \bu_k, \bu_t) & =  \left(-\bbf
  \cdot \bv+\bbf_k\cdot \bu_k + \bbf_t\cdot \bu_t\right)^2
  \end{aligned}
 \]
 where $\{k,t\}=\{0,1,2,3\}\setminus\{i,j\}$. The quadratic form for both networks taken together is
 $$q(\bu_0, \bu_1, \bu_2, \bu_3) = \inf_\bv \{q_1(\bv, \bu_i, \bu_j) +
 q_2(\bv, \bu_k, \bu_t)\}.$$ The optimality conditions are
 \[
  \bbf \cdot \bv = - \frac{1}{2} ( \bbf_i \cdot \bu_i + \bbf_j \cdot
  \bu_j)+
  \frac{1}{2}  (\bbf_k \cdot \bu_k+ \bbf_t \cdot \bu_t)
 \]
 which yield the desired result
 \[
  \begin{aligned}
   q(\bu_0, \bu_1, \bu_2, \bu_3) = \frac{1}{2} \M{\sum_{i=0}^3 \bbf_i \cdot
   \bu_i}^2.
  \end{aligned}
 \]
 \end{proof}

The following Theorem is the main result of this section. We use the
previous results for three and four terminal networks to prove the
result in the general case of $p$-terminal networks by induction,
following the approach of Camar-Eddine and Seppecher
\cite{camareddine_seppecher2003}.

\begin{theorem}
 \label{thm:rankone}
 Let $\bbf_i$ and $\bx_i$, $i=1,\ldots,p$ be a balanced system of forces in
 $\real^2$. There is a purely elastic $p$ terminal planar network with a
 force response proportional to $(\bbf_1^T,\ldots,\bbf_p^T)^T$ and with
 internal nodes in an $\eps-$neighborhood of the convex hull of the
 terminals and avoiding a finite number of points.
\end{theorem}
\begin{proof}
We use an induction argument in the number of terminals as in
\cite[Theorem 5]{camareddine_seppecher2003}, which we reproduce here for
completeness. The $p=2$ case corresponds to a single spring with same
direction as the forces. The cases $p=3$ and $p=4$ are proved in
\lemref{lem:threeterm}, \lemref{lem:threeterm-1D} and
\lemref{lem:fourterm}. Assume for the induction argument that the
theorem holds for any $t<p$ terminals. Let $r$ be the integer part of
$p/2$, and let $\by$ be a node distinct from the terminals
$\bx_0,\ldots,\bx_p$. Then there are two forces $\bbf$ and $\bbf'$ such
that both families
\[
 \begin{aligned}
 & ((\bbf,\by), (\bbf_1 + \bbf',\bx_1), (\bbf_2,\bx_2),\ldots, (\bbf_{r+1},\bx_{r+1}))~\text{and}\\
 & ((-\bbf,\by), (-\bbf',\bx_1), (\bbf_{r+1},\bx_{r+1}), \ldots , (\bbf_p,\bx_p)).
 \end{aligned}
\]
are balanced systems of forces, as can easily be seen by taking
$\bbf'=-(\bbf+\bbf_1+\bbf_2+\ldots+\bbf_{r+1})$ and choosing $\bx_1$ as the origin of
coordinates.  These families have $r+1$ and $p-r+1$ terminal nodes, and
both have less than $p$ terminals when $p>4$. By the induction
hypothesis there are rank one networks with associated quadratic forms
\[
\begin{aligned}
 q'(\bv,\bu_1,\ldots,\bu_{r+1}) &= \M{\bbf\cdot \bv + \bbf' \cdot \bu_1 + \sum_{i=1}^{r+1} \bbf_i \cdot \bu_i}^2,~\text{and}\\
 q''(\bv,\bu_{r+1},\ldots,\bu_p) &= \M{-\bbf\cdot \bv - \bbf' \cdot \bu_1 + \sum_{i=r+1}^p \bbf_i \cdot \bu_1 }^2.
\end{aligned}
\]
The quadratic form of both networks together is
\[
 \tq(\bu_1,\bu_2,\ldots,\bu_p) = \inf_{\bv \in \real^2} q'(\bv,\bu_1,\ldots,\bu_{r+1}) + q''(\bv,\bu_{r+1},\ldots,\bu_p).
\]
The optimality conditions are:
\[
 \bbf\cdot \bv = -\frac{1}{2} \M{\bbf' \cdot \bu_1 + \sum_{i=1}^{r+1} \bbf_i \cdot \bu_i} + \frac{1}{2} \M{\bbf'\cdot \bu_1 -\sum_{i=r+1}^p \bbf_i \cdot \bu_1}.
\]
Thus the quadratic form of both networks is a rank one with
\[
  \tq(\bu_1,\ldots,\bu_p) = \frac{1}{2} \M{\sum_{i=1}^p \bbf_i \cdot \bu_i}^2.
\]
Notice that the additional point $\by$ can be picked inside an
$\eps-$neighborhood of the convex hull of the terminal nodes and
avoiding a finite number of points. Also we can use virtual crossings
and trusses to make the network planar (see \secref{sec:transf}).
\end{proof}

\begin{remark}
  \label{rem:planar}
  The networks in Lemmas \ref{lem:threeterm}, \ref{lem:threeterm-1D},
  \ref{lem:fourterm'}, \ref{lem:fourterm}, and \thmref{thm:rankone} may
  have crossing springs so are not strictly planar. To make them planar
  is suffices to convert all spring crossings with non-zero angle to a
  node as is done in \cite[Example 3.15]{milton_seppecher2008}.
  Zero-angle crossings can be eliminated by replacing springs with
  simple trusses \cite[Example 3.2]{milton_seppecher2008}. These network
  transformations are discussed in \secref{sec:transf}.
\end{remark}

\subsection{Networks in $\real^3$ with rank one static response matrices}
The construction of these networks is essentially due to Camar-Eddine
and Seppecher \cite{camareddine_seppecher2003}. We first complete their
construction
of rank one four terminal networks to include some degenerate cases
which correspond to planar networks (\lemref{lem:rankone3d}).  Then
following Camar-Eddine and Seppecher \cite{camareddine_seppecher2003} we
use induction (\thmref{thm:rankone_gen}) to derive rank one networks
with an arbitrary number of terminals.
\begin{lemma}
\label{lem:rankone3d}Let $\bbf_0$, $\bbf_1$, $\bbf_2$, $\bbf_3$
be a set of balanced forces at the nodes $\bx_0$, $\bx_1$,
$\bx_2$, $\bx_3$ in $\real^3$. There is a purely elastic rank one network
with force response proportional to $(\bbf_0^T, \bbf_1^T,
\bbf_2^T,\bbf_3^T)^T$. Moreover the internal nodes can be chosen within an $\eps-$neighborhood of the convex hull of the terminals and avoiding a finite number of points.
\end{lemma}
\begin{proof}
As in the planar case there are two cases depending on the value of $$r \equiv
\rank[ \bbf_1, \bbf_2, \bbf_3, \bx_1 - \bx_0, \bx_2 - \bx_0, \bx_3 - \bx_0].$$
The construction for $r=3$ is given in \cite[Lemma
5]{camareddine_seppecher2003}. When $r\leq 2$ the network is planar, so the
result follows from \thmref{thm:rankone}. 
\end{proof}

\begin{theorem}
 \label{thm:rankone_gen}
 Let $\bbf_i$ and $\bx_i$, $i=1,\ldots,p$ be a balanced system of forces
 in $\real^3$. There is a purely elastic $p$ terminal network with a
 force response proportional to $(\bbf_1^T,\ldots,\bbf_p^T)^T$. Moreover
 the internal nodes can be chosen within an $\eps-$neighborhood of the
 convex hull of the terminals and avoiding a finite number of points.
\end{theorem}
\begin{proof}
The result follows from an induction argument similar to that of
\citet[Theorem 5]{camareddine_seppecher2003}. See also the proof of
\thmref{thm:rankone}.
\end{proof}
\section{Characterization of the dynamic response function}
\label{sec:dynamic}
To fully characterize the dynamic response matrices, we take a
function $\bW(\om)$ as in \lemref{lem:dynamic} and show that we can construct a
network that has $\bW(\om)$ as its frequency response. The construction
relies on the static case (\thmref{thm:charstat}) and the existence of a
network of springs and masses with rank one response that has exactly one
prescribed resonance (\lemref{lem:reso_rankone}). Both networks
in $\real^3$ and planar can be constructed.
\begin{theorem}
\label{thm:chardyn}
Let $\bW(\om)$ be a matrix valued function of $\om$ satisfying the properties of
\lemref{lem:dynamic}. Then for any choice of terminal node positions,
there is a network (either planar or in $\real^3$) of springs and masses
with $\bW(\om)$ as its response function.  Moreover the internal nodes of
such a network can be chosen to avoid a finite number of positions and
within an $\eps-$neighborhood of the convex hull of the terminals.
\end{theorem}
\begin{proof}
It is convenient to rewrite \eqref{eq:wom} as,
\begin{equation}
 \label{eq:wom2}
 \bW(\om) = \bW(0)  - \om^2 \bM + \sum_{i=1}^p \bC^{(i)} \frac{\om^2}{\om_i^2 (\om^2-\om_i^2)}.
\end{equation}
Since $\bW(0)$ has the properties of \lemref{lem:static}, by
\thmref{thm:charstat} there is a (static) network of springs that has
$\bW(0)$ as its response matrix. Since the $\bC^{(i)}$ are positive semidefinite we can use the spectral decomposition to write
\[
 \bC^{(i)} = \sum_{j=1}^{n_i} \lambda^{(i)}_j \bc^{(i)}_j (\bc^{(i)}_j)^T,
 ~\text{for $i=1,\ldots,p$},
\]
where the $(\lambda^{(i)}_j, \bc^{(i)}_j)$ are the eigenpairs of $\bC^{(i)}$
with $\lambda^{(i)}_j>0$ and $n_i = \rank(\bC^{(i)})$.  By
\lemref{lem:reso_rankone} we can construct a network with response
function 
\[
\M{\sqrt{\lambda_j^{(i)}} \om_i^{-1}\bc_j^{(i)}}\M{\sqrt{\lambda_j^{(i)}} \om_i^{-1}\bc_j^{(i)}}^T \frac{\om^2}{\om^2 - \om_i^2}.
\]
By the superposition principle there is a network with response the sum
appearing in \eqref{eq:wom2}.  Finally to obtain the $-\om^2\bM$ term,
simply endow the terminal node $\bx_i$ with a mass equal to the $id-$th
diagonal element of $\bM$ (The mass of node $i$ is repeated $d$ times in
the $nd \times nd$ matrix $\bM$). To obtain a planar network, simply replace
all spring crossings by ``virtual crossings'' (see \secref{sec:transf}).
\end{proof}

\subsection{Rank one response matrices with resonance}
The following lemma shows how to design a network with arbitrary rank
one response and one single resonance at a prescribed frequency. The
network we construct has a purely dynamic response as it has a zero
response matrix in the static case $\om=0$. The construction is valid for both
networks in $\real^3$ and planar networks.
\begin{lemma}
 \label{lem:reso_rankone}
 Let $\bx_i$ and $\bbf_i$, $i=1,\ldots,n$, be arbitrary points and forces and
 $\om_0\neq 0$ a given finite resonance frequency. There is a network with
 terminals $\bx_i$ composed of springs and masses with rank one response
 function 
 \[
   \bW(\om) = \bbf \bbf^T \frac{\om^2}{\om^2 - \om_0^2},
   ~\text{where}~
   \bbf^T = (\bbf_1^T,\bbf_2^T,\ldots,\bbf_n^T).
 \]
Moreover the internal nodes of such a network can be chosen to avoid a finite
number of positions and within an $\eps-$neighborhood of the convex hull of
the terminals.
\end{lemma}
\begin{proof}
 Let $\bx_{n+1}$ and $\bx_{n+2}$ be two distinct nodes and choose the
 forces $\bbf_{n+1}$ and $\bbf_{n+2}$ so that the system
 $(\bx_i,\bbf_i)$, $i=1,\ldots,n+2$ is balanced. Take for example a
 force $\bbf_{n+2} \neq \bzero$ in the line ($d=2$) or plane ($d=3$)
 \[
   (\bx_{n+2} - \bx_{n+1}) \cross \bbf_{n+2} = - \sum_{i=1}^n (\bx_i
   -\bx_{n+1}) \cross \bbf_i,
 \]
 and choose $\bbf_{n+1}$ such that $\sum_{i=1}^{n+2} \bbf_i = \bzero$.  Then
 by \thmref{thm:rankone} there is a rank one network with force response
 proportional to $(\bbf^T,\bbf_{n+1}^T,\bbf_{n+2}^T)^T$. Attach a mass $m$
 to nodes $\bx_{n+1}$ and $\bx_{n+2}$. The spring constants in the network
 can be rescaled (\remref{rem:rescale}) so that the equations of motion are
 \[
  \begin{bmatrix} \bw_B\\ m \om^2 \bu_I \end{bmatrix}
  =
  \begin{bmatrix} \bbf\\ \ba \end{bmatrix}
  \begin{bmatrix} \bbf^T  & \ba^T \end{bmatrix}
  \begin{bmatrix} \bu_B\\ \bu_I \end{bmatrix},
 \]
 where $\ba^T =
 (\bbf_{n+1}^T, \bbf_{n+2}^T) \neq \bzero$, the displacements $\bu_B$ and $\bu_I$
 are respectively at the ``boundary''
 nodes $\bx_1,\ldots,\bx_n$ and the ``interior'' nodes
 $\bx_{n+1},\bx_{n+2}$. Then solving the system for the resulting
 forces $\bw_B$ at the ``boundary'' nodes we get,
 \[
  \bw_B = \bbf \bbf ^T \bu_B \M{1 + \frac{ \norm{\ba}^2}{m\om^2 -
  \norm{\ba}^2}} =\bbf \bbf ^T \bu_B
  \frac{\om^2}{\om^2 - \norm{\ba}^2/m}.
 \]
 Finally choose the mass $m = \norm{\ba}^2/\om_0^2$. The position of the
 internal nodes $\bx_{n+1}$ and $\bx_{n+2}$ is flexible and by
 \thmref{thm:rankone} so is that of any interior nodes in the rank one
 network involved in the construction.
\end{proof}

\appendix
\section{Stability to small perturbations}
\label{app:pert}
We show that the response function of an elastodynamic network is stable
to changes in the network, which could come from either modifying the
spring constants of existing springs or possibly adding springs with small
spring constants between any two nodes in the network. However we do not
allow springs to be deleted from the network. We first show
stability of the response of static networks and then stability for the
response function of elastodynamic networks. 

\subsection{Stability in the static case}
Let $\bA$ be the response matrix of an elastic network with all nodes
considered as terminals and $\bW$ be the response matrix at the terminal
nodes as given by \eqref{eq:schur}. If we add or modify (but not
delete) springs then the new response with all nodes considered as
terminals is $\bA + \eps\bE$, $\eps>0$, and its response at the terminals
$\bW(\eps)$. We prove the following result
\begin{lemma}
 \label{lem:stabstat}
 Let $\eps>0$. As $\eps \to 0$, we have $\bW(\eps) \to \bW$.
\end{lemma}

The stability result for the static case may seem surprising at first
because the pseudo-inverse we used to find the response at the terminals
\eqref{eq:schur} is not continuous (see e.g. \cite[\S 5.5.5]{gvl}).
However \lemref{lem:floppy} guarantees that the instabilities are
controlled as they remain (roughly speaking) in $\nullspace(\bA_{II})$.
Before showing \lemref{lem:stabstat} we need to establish the following
relation between the floppy modes of the perturbed and unperturbed
stiffness matrices.

\begin{lemma}
 \label{lem:floppypert}
 For all $\eps>0 $ sufficiently small, $\nullspace(\bA_{II} +
 \eps\bE_{II}) \subset \nullspace(\bA_{II})$. Moreover
 $\nullspace(\bA_{II} + \eps \bE_{II})$ is {\em independent} of $\eps$, and
 depends only on the connectivity of the new network. In other words if a
 network is perturbed by adding springs or modifying existing springs,
 then a floppy mode of the perturbed network must be a floppy mode of
 the unperturbed network.
\end{lemma}
\begin{proof}
If $\bu_I$ is a floppy mode of the perturbed network then:
 \[
  \begin{aligned}
  0 &= \bu_I^T (\bA_{II} + \eps \bE_{II}) \bu_{I} 
  = \begin{bmatrix} \bzero & \bu_I^T \end{bmatrix} (\bA + \eps \bE) \begin{bmatrix} \bzero \\\bu_I \end{bmatrix}\\
  &= \sum_{\substack{\text{old springs}\\i\in I,\;j \in B \cup I}} 
   (k_{i,j} + \eps l_{i,j}) \M{(\bu_i-\bu_j) \cdot \frac{\bx_i -
   \bx_j}{\norm{\bx_i-\bx_j}}}^2\\
  &+ \sum_{\substack{\text{new springs}\\  i\in I,\;j \in B \cup I}} 
   \eps l_{i,j} \M{(\bu_i-\bu_j) \cdot \frac{\bx_i -
   \bx_j}{\norm{\bx_i-\bx_j}}}^2.
  \end{aligned}
 \]
Since the perturbed network is a spring network, the new spring
constants should be positive and all the terms in the sums above
vanish, a condition which is independent of $\eps$. Since all the terms
in the sums above vanish it follows that $\bu_I \in
\nullspace(\bA_{II})$:
\[
 \bu_I^T \bA_{II} \bu_I = \sum_{\substack{\text{old springs}\\i\in I,\;j \in B
 \cup I}} k_{i,j} \M{(\bu_i-\bu_j) \cdot \frac{\bx_i -
   \bx_j}{\norm{\bx_i-\bx_j}}}^2 = 0.
\]
\end{proof}
We are now ready to prove stability for the static case.
\begin{proof} (of \lemref{lem:stabstat}) Let $\eps>0$ be sufficiently
small. By \lemref{lem:floppypert} it is possible to find a unitary
matrix $[\bU,\bV,\bW]$ {\em independent} of $\eps$ such that
$\range([\bV,\bW]) = \nullspace(\bA_{II})$ and $\range(\bW) =
\nullspace(\bA_{II} + \eps \bE_{II})$. Writing $\bA_{II} + \eps
\bE_{II}$ in the new basis gives,
\begin{equation}
 \bA_{II} + \eps \bE_{II} = [ \bU, \bV, \bW ] 
 \begin{bmatrix}
 \btA + \eps \btE_1 & \eps \btE_2 & \bzero \\
 \eps \btE_2^T & \eps \btE_3 & \bzero \\
 \bzero & \bzero & \bzero
 \end{bmatrix}
 [\bU,\bV,\bW]^T,
 \label{eq:uvw}
\end{equation}
where $\btA \equiv \bU^T \bA_{II} \bU$. Because of our choice of basis
both $\btA$ and the non-zero block in \eqref{eq:uvw} must be invertible
and symmetric positive definite.  The inverse of this block is
\begin{equation}
\begin{aligned}
 &\begin{bmatrix}
  \btA + \eps \btE_1 & \eps \btE_2\\
  \eps \btE_2^T & \eps \btE_3
 \end{bmatrix}^{-1} \\
 &\qquad=
 \begin{bmatrix}
 \btA(\eps)^{-1}(\bI + \eps \btE_2 \btE(\eps)^{-1} \btE_3^{-1}
 \btE_2^T \btA(\eps)^{-1})
 & -\btA(\eps)^{-1} \btE_2 \btE(\eps)^{-1} \btE_3^{-1}\\
 -\btE(\eps)^{-1} \btE_3^{-1} \btE_2^T \btA(\eps)^{-1}
 & \eps^{-1} \btE(\eps)^{-1} \btE_3^{-1} 
 \end{bmatrix}\\
 &\qquad= 
 \begin{bmatrix}
   \btA^{-1} + \eps \bG_1 + o(\eps) & - \btA^{-1} \btE_2 \btE_3^{-1} +
   \eps \bG_2 + o(\eps)\\
   -\btE_3^{-1} \btE_2^T \btA^{-1} + \eps \bG_2^T + o(\eps) &
   \eps^{-1} \btE_3^{-1} + \bG_3 + o(1)
 \end{bmatrix},
 \end{aligned}
\end{equation}
where $\btE(\eps) = (\bI - \eps \btE_3^{-1} \btE_2^T \btA(\eps)^{-1}
\btE_2)$ and $\btA(\eps) = \btA + \eps \bE_1$.  Notice that $\btE_3$ is
a submatrix of a symmetric positive definite matrix and thus must be
invertible. The second equality comes from the standard perturbation
formula for the inverse (see e.g. \cite[\S 2.3.4]{gvl}) and the matrices
$\bG_1,\bG_2,\bG_3$ are independent of $\eps$. 
We now examine the response of the perturbed matrix. Since the first
term in \eqref{eq:schur} is linear in $\bE$, it is stable to
perturbations. Now the negative of the second term in \eqref{eq:schur}
can be written as:
\begin{equation}
 \label{eq:abiaiiaib}
 \begin{aligned}
 &(\bA_{BI} + \eps \bE_{BI}) (\bA_{II} + \eps\bE_{II})^\dagger ( \bA_{IB}
 + \eps \bE_{IB}) = 
 \bA_{BI} (\bA_{II} + \eps\bE_{II})^\dagger \bA_{IB}\\
 &+\eps \bE_{BI} (\bA_{II} + \eps\bE_{II})^\dagger \bA_{IB}
 + \eps \bA_{BI} (\bA_{II} + \eps\bE_{II})^\dagger \bE_{IB}
 + \eps^2 \bE_{BI}  (\bA_{II} + \eps\bE_{II})^\dagger \bE_{IB}.
 \end{aligned}
\end{equation}
Moreover in the basis $[\bU,\bV,\bW]$ the pseudo-inverse becomes,
\[
(\bA_{II} + \eps\bE_{II})^\dagger
 = [\bU,\bV]
 \begin{bmatrix}
  \btA + \eps \btE_1 & \eps \btE_2\\
  \eps \btE_2^T & \eps \btE_3
 \end{bmatrix}^{-1}
 [\bU,\bV]^T.
\]
Using \lemref{lem:floppy} and \lemref{lem:floppypert}, we have $\bA_{BI}
[ \bV, \bW ] =[\bzero,\bzero]$ and $(\bA_{BI} + \eps \bE_{BI}) \bW =
\bzero$. The leading order asymptotics of each of the terms in
\eqref{eq:abiaiiaib} are
\[
\begin{aligned}
\bA_{BI} (\bA_{II} + \eps\bE_{II})^\dagger \bA_{IB}
&= \bA_{BI} \bU \btA^{-1} \bU^T \bA_{IB} + \cO(\eps) = \bA_{BI}
\bA_{II}^\dagger \bA_{IB} + \cO(\eps),\\
 \eps \bE_{BI} (\bA_{II} + \eps\bE_{II})^\dagger \bA_{IB}
 & =  \eps \bE_{BI} (\bU \btA^{-1} \bU^T - \bV \btE_3^{-1} \btE_2^T
 \btA^{-1}\bU^T)\bA_{IB} + \cO(\eps^2),\\
 \eps \bA_{BI} (\bA_{II} + \eps\bE_{II})^\dagger \bE_{IB}
 & = \eps \bA_{BI}  (\bU \btA^{-1} \bU^T - \bU \btA^{-1} \btE_2
 \btE_3^{-1} \bV^T) \bE_{IB} + \cO(\eps^2),\\
 \eps^2 \bE_{BI}  (\bA_{II} + \eps\bE_{II})^\dagger \bE_{IB}
 & = \eps \bE_{BI} \bV \btE_3^{-1} \bV^T \bE_{IB} + \cO(\eps^2).
\end{aligned}
\]
which proves the desired result:
\[
(\bA_{BI} + \eps \bE_{BI}) (\bA_{II} + \eps\bE_{II})^\dagger ( \bA_{IB}
 + \eps \bE_{IB})  = \bA_{BI} \bA_{II}^\dagger \bA_{IB} + \cO(\eps).  \]
\end{proof}

\subsection{Stability in the dynamic case}
As in the static case, we deal only with network perturbations that
modify existing springs or add new springs, but excluding 
spring deletions. Denote by $\bK$ the response when all the nodes
are terminals and let $\eps>0$ be sufficiently small so that $\bK +
\eps\bE$ is the response of the perturbed network. We show the following
result.

\begin{lemma}
 Partition the interior nodes $I$ into massless nodes $L$ and nodes with
 mass $J$, as in  \lemref{lem:schurdyn}. Let $\om$ be a frequency such
 that $\om^2$ is not an eigenvalue of $\bM_{JJ}^{-1/2}\btK_{JJ}
 \bM_{JJ}^{-1/2}$, and $\bK + \eps\bE$ be the response of the perturbed
 network where we allow for new springs or changes in the spring
 constants, but no spring deletions. Then as $\eps \to 0$,
 \[
  \bW(\om;\eps) = \bW(\om)  + \cO(\eps),
 \]
 where $\bW(\om;\eps)$ (resp. $\bW(\om)$) is the response at the terminal nodes of the
 perturbed (resp. unperturbed) network at frequency $\om$.
\end{lemma}
\begin{proof}
 By \lemref{lem:stabstat} the matrix $\btK$ (the
response matrix of the network with terminals $B \cup J$, see
\eqref{eq:tk}) is stable to such spring perturbations, meaning that the
response of the perturbed network at the nodes $B \cup J$ satisfies
$\btK(\eps) = \btK + \eps \btE + o(\eps)$, for some matrix $\btE$
independent of $\eps$. Since both the perturbed and unperturbed
responses are symmetric, the Wielandt-Hoffman theorem (see e.g. \cite[\S
8.1.2]{gvl}) implies that there is a reordering of the eigenvalues
$\om_i^2(\eps)$ of
$\btK_{JJ}(\eps)$ such that 
\[
 |\om_i^2(\eps) -\om_i^2| \leq \eps \|\btE\|_F,
\] 
with $\om_i^2$ being the eigenvalues of $\btK_{JJ}$ and
$\norm{\;\cdot\;}_F$ denoting the Frobenius matrix norm. Therefore if
$\eps$ is sufficiently small and $\om$ is not an eigenvalue of
$\bM_{JJ}^{-1/2} \btK_{JJ} \bM_{JJ}^{-1/2}$ then $\om$ is not an
eigenvalue of $\bM_{JJ}^{-1/2} \btK_{JJ}(\eps) \bM_{JJ}^{-1/2}$ either
and the matrix $\btK(\eps) - \om^2 \bM_{JJ}$ is invertible. Thus using
\eqref{eq:schurdyn} and the standard perturbation formula for the
inverse,  it is possible to show that the response at the
terminals of the perturbed matrix is $\bW(\om;\eps) = \bW(\om) +
\cO(\eps)$.
\end{proof}
\section{Eliminating floppy modes by adding springs}
\label{app:floppy}
We use the stability results from the previous Appendix to show that if
a network has floppy modes then there is a network with no floppy modes
and a response function arbitrarily close to that of the original
network. Some examples of floppy modes (for planar networks) are given in
\figref{fig:floppy}. Our strategy to remove floppy modes is to connect
all nodes (be them terminal or interior nodes) of the network by springs
with small spring constants, thus creating a {\em complete graph} with
the nodes $I \cup B$ and springs as edges. By \lemref{lem:stabstat} the
response of the new network can be made arbitrarily close to the
response of the unperturbed network.

Let $\bA$ be the response of the network where all nodes are terminals
and let $\bA+\eps\bE$ be response when we have added all these new
springs. Then  if $\bu_I$ is a floppy mode of the new network,
proceeding as in \lemref{lem:floppypert} gives
\[
  \begin{aligned}
  0 &= \bu_I^T (\bA_{II} + \eps \bE_{II}) \bu_{I} 
  = \begin{bmatrix} \bzero & \bu_I^T \end{bmatrix} (\bA + \eps \bE) \begin{bmatrix} \bzero \\\bu_I \end{bmatrix}\\
  &= \sum_{\substack{\text{old springs}\\i\in I,\;j \in B \cup I}} 
   (k_{i,j} + \eps l_{i,j}) \M{(\bu_i-\bu_j) \cdot \frac{\bx_i -
   \bx_j}{\norm{\bx_i-\bx_j}}}^2\\
  &+ \sum_{\substack{\text{new springs}\\  i\in I,\;j \in B \cup I}} 
   \eps l_{i,j} \M{(\bu_i-\bu_j) \cdot \frac{\bx_i -
   \bx_j}{\norm{\bx_i-\bx_j}}}^2.
  \end{aligned}
 \]
 This is equivalent to saying that
 \begin{equation}
  \forall i \in I ~\text{and}~  j \in B \cup I, \quad (\bu_i-\bu_j)\cdot
  (\bx_i-\bx_j) = 0.
  \label{eq:zero}
 \end{equation}
 Assume that we are working in $d$ dimensions and that we have $d+1$
 nodes $\bx_1,\ldots,\bx_{d+1}$ that form a non-degenerate triangle
 ($d=2$) or tetrahedron ($d=3$) and for which $\bu_1 = \bu_2 = \cdots =
 \bu_{d+1} = \bzero$. Then since every interior node $\by$ is connected
 to the $\bx_1,\ldots,\bx_{d+1}$, equation \eqref{eq:zero} implies that
 $\bv \cdot (\by - \bx_i) = 0$, $i=1,\ldots,d+1$, where $\bv$ is the
 displacement associated with $\by$. Since these special ``anchor''
 nodes are non-degenerate, $\rank[\bx_1 - \by, \bx_2 - \by,\ldots
 ,\bx_{d+1} - \by] = d$ and we must have $\bv = \bzero$. Repeating this
 for every interior node we get $\bu_I = \bzero$, and so the network does
 not have any floppy modes.

 We now need to show which networks have ``anchor'' nodes. Clearly if
 the network has $d+1$ terminal nodes forming a non-degenerate triangle
 (in $d=2$) or tetrahedron (in $d=3$), then the network does not have
 any floppy modes, since the terminal nodes do not move.

 If we are in $d=2$ dimension and the terminal nodes do not form a
 non-degenerate triangle, then all terminals must lie on a line. Since
 the network has at least two terminal nodes (otherwise we cannot
 balance forces), we can add two interior nodes with a truss as in
 \secref{sec:transf} without changing the response. Let $\bx_1,\bx_2$ be
 two terminal nodes and $\by$ be one of the interior nodes of the truss,
 with associated displacements $\bu_1=\bu_2 = \bzero$ and $\bv$. Then
 condition \eqref{eq:zero} implies that $\bv \cdot (\bx_1 - \by) = \bv
 \cdot (\bx_2 - \by) = 0$, i.e. $\bv = \bzero$. Thus the nodes
 $\bx_1,\bx_2,\by$ form an ``anchor'' and the network does not have any
 floppy modes.

 If we are in $d=3$ dimension and we cannot form a non-degenerate
 tetrahedron from the terminal nodes, then the terminals must lie on a
 plane. Assume further that the terminals do not lie on a line, we shall
 deal with this case later. Let $\bx_1,\bx_2,\bx_3$ be three terminal
 nodes forming a triangle. Then replacing e.g. the spring between
 $\bx_1$ and $\bx_2$ by a truss (as in \secref{sec:transf}), we
 introduce three new interior nodes and at least one of them $\by$ is
 not in the plane where $\bx_1,\bx_2,\bx_3$ lie. Let $\bv$ be the
 displacement of $\by$ and $\bu_1 = \bu_2 = \bu_3= \bzero$ be the
 displacements of the boundary nodes.  Then condition \eqref{eq:zero}
 implies that $\bv \cdot (\bx_i - \by) = 0$ for $i=1,2,3$. Since
 $\rank[\bx_1-\by,\bx_2-\by,\bx_3-\by] = 3$, we must have $\bv =
 \bzero$. Therefore the nodes $\bx_1,\bx_2,\bx_3,\by$ form an ``anchor''
 and the network does not have any floppy modes.

 If we are in $d=3$ dimension and all the terminal nodes lie on a line
 then every interior node $\by$ forms a triangle with two terminal nodes
 $\bx_1,\bx_2$. In this case condition \eqref{eq:zero} means that the
 displacement $\bv$ of node $\by$ is orthogonal to the plane formed by
 $\bx_1,\bx_2,\by$, and in particular to the axis where all terminals
 lie. Thus in this case the floppy modes cannot be eliminated by adding
 springs or interior nodes, as any additional interior node is in this
 situation as well. This corresponds to rotations of the network around
 the axis where all the terminals lie.

\begin{figure}
 \begin{center}
 \begin{tabular}{c@{\hspace{2em}}c}
\includegraphics{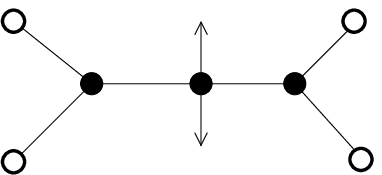}%
&%
\includegraphics{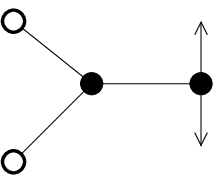}%
 \end{tabular}
 \end{center}
 \caption{Two types of floppy modes. The terminals are white circles
 and the interior nodes are in black. The direction in which the node
 can move with zero force is given with arrows.}
 \label{fig:floppy}
\end{figure}

\section*{Acknowledgements}
The authors wish to thank Pierre Seppecher for helpful conversations.
The authors are grateful for support from the National Science Foundation
through grant DMS-0707978.

\bibliographystyle{abbrvnat}
\bibliography{elnetbib}

\end{document}